\documentclass[11pt]{article}

\usepackage[margin=1in]{geometry}
\usepackage{amsmath,amssymb,amsthm,mathtools}
\usepackage{algorithm}
\usepackage{algpseudocode}
\usepackage{enumitem}
\usepackage{microtype}
\usepackage{hyperref}
\hypersetup{hidelinks}

\newtheorem{theorem}{Theorem}
\newtheorem{lemma}[theorem]{Lemma}

\theoremstyle{definition}
\newtheorem{definition}[theorem]{Definition}
\newtheorem{remark}[theorem]{Remark}

\newcommand{\OPT}{\mathrm{OPT}}
\newcommand{\ALG}{\mathrm{ALG}}
\newcommand{\cc}{\operatorname{cc}}

\newcommand{\cutG}[1]{c_G(\delta_G(#1))}

\title{An \(O(\log n)\)-Approximation for Three-Terminal\\
Reachability-Preserving Minimum Edge Cut}

\author{Qi Duan, Carnegie Mellon University}
\date{}

\begin{document}

\maketitle

\begin{abstract}
In the three-terminal Reachability-Preserving Minimum Edge Cut problem,
the input is an undirected edge-weighted graph with terminals
\(s_1,s_2,t\). The objective is to delete a minimum-cost set of edges
that separates \(t\) from both \(s_1\) and \(s_2\), while preserving
connectivity between \(s_1\) and \(s_2\).

We give a polynomial-time \(O(\log n)\)-approximation algorithm. The
algorithm uses a probabilistic distribution of cut-dominating
decomposition trees. A direct transfer of a connected tree solution to
the original graph is not valid because a connected tree cluster may
induce a disconnected vertex set in the graph. We overcome this
obstruction by expanding every rooted tree cluster into the connected
components it induces in the original graph. These components form a
node-weighted auxiliary graph. A minimum node-weighted path in this
auxiliary graph produces a connected feasible source side.

The main structural observation is that the total graph-boundary cost
of all connected components of a rooted tree cluster is no greater
than the capacity of the corresponding tree edge. This permits the
auxiliary path to be compared with a tree cut separating an optimal
preserved \(s_1\)-\(s_2\) path from \(t\). Combining this comparison
with the expected \(O(\log n)\) cut distortion of the decomposition
trees proves the approximation guarantee.
\end{abstract}

\section{Introduction}

The three-terminal Reachability-Preserving Minimum Edge Cut problem,
abbreviated RPMEC, combines a separation requirement with a protected
connectivity requirement. Given an undirected graph with terminals
\(s_1,s_2,t\), one must separate \(t\) from both source terminals while
ensuring that the two source terminals remain mutually reachable.

The connectivity requirement distinguishes RPMEC from ordinary
minimum cut. A cut of small capacity may place \(s_1\) and \(s_2\) on
the same side while leaving them in different connected components.
Consequently, methods that preserve only cut capacities do not
automatically preserve feasibility.

Previous work obtained an \(O(\sqrt n)\)-approximation through a
root-linear or modular-surrogate formulation~\cite{rpmecCurrent}. A
natural question is whether a logarithmic approximation is possible.

We answer this question affirmatively. Our algorithm uses the
probabilistic cut-tree decomposition framework of
R\"acke~\cite{racke2008} and its capacity-mapping interpretation
\cite{andersenfeige}. These results provide a polynomial-size
distribution over decomposition trees such that every support tree
dominates all graph cuts, while each fixed graph cut has only
\(O(\log n)\) expected tree-cut capacity.

The central technical difficulty is that a tree cluster need not induce
a connected subgraph of the original graph. Our solution is to replace
each rooted tree cluster by all connected components of its induced
graph. We create one auxiliary node for each such component, weighted
by its graph-boundary cost. Two auxiliary nodes are adjacent when the
corresponding components intersect or are joined by a graph edge.

A shortest auxiliary path from a component containing \(s_1\) to a
component containing \(s_2\) yields a connected union in the original
graph. The total boundary of this union is at most the auxiliary-path
weight.

The key accounting identity is the following. If \(C_e\) is the graph
vertex set below a rooted tree edge \(e\), then
\[
\sum_{K\in\cc(G[C_e])}\cutG{K}
=
\cutG{C_e}.
\]
Since the decomposition-tree edge capacity dominates \(\cutG{C_e}\),
all components generated by \(e\) can be charged to that single tree
edge.

This component-budget property allows us to compare the auxiliary
shortest path with any tree cut separating \(t\) from a fixed
\(s_1\)-\(s_2\) graph path. Choosing such a path inside an optimal
RPMEC source side completes the \(O(\log n)\) analysis.

Our main result is as follows.

\begin{theorem}\label{thm:main-intro}
Weighted undirected three-terminal RPMEC with nonnegative rational edge
costs admits a polynomial-time \(O(\log n)\)-approximation algorithm.
\end{theorem}

\section{Problem Definition and Source-Side Formulation}

Let
\[
G=(V,E)
\]
be an undirected graph with a nonnegative rational edge-cost function
\[
c:E\rightarrow\mathbb{Q}_{\ge 0}.
\]
Let \(s_1,s_2,t\in V\) be distinct terminals.

For \(U\subseteq V\), define
\[
\delta_G(U)
=
\{uv\in E: |\{u,v\}\cap U|=1\},
\]
and let
\[
\cutG{U}
=
\sum_{e\in\delta_G(U)}c(e).
\]

\begin{definition}[Three-terminal RPMEC]
A feasible RPMEC solution is an edge set \(F\subseteq E\) such that,
in \(G-F\),
\begin{enumerate}[label=(\roman*)]
    \item \(s_1\) and \(s_2\) are connected; and
    \item \(t\) is disconnected from both \(s_1\) and \(s_2\).
\end{enumerate}
The objective is to minimize
\[
c(F)=\sum_{e\in F}c(e).
\]
\end{definition}

The problem has a useful equivalent formulation in terms of a
connected source side.

\begin{lemma}[Connected source-side formulation]
\label{lem:source-side}
The optimum RPMEC value equals
\[
\OPT
=
\min
\left\{
    \cutG{U}:
    s_1,s_2\in U,\;
    t\notin U,\;
    G[U]\text{ is connected}
\right\}.
\tag{1}\label{eq:source-side}
\]
\end{lemma}

\begin{proof}
Let \(U\subseteq V\) satisfy the conditions in
\eqref{eq:source-side}. Deleting \(\delta_G(U)\) separates every
vertex of \(U\) from \(V\setminus U\), and therefore separates
\(t\notin U\) from \(s_1,s_2\in U\). Since \(G[U]\) is connected,
\(s_1\) and \(s_2\) remain connected. Hence \(\delta_G(U)\) is a
feasible RPMEC solution.

Conversely, let \(F\) be any feasible RPMEC solution. Let \(U\) be the
connected component of \(G-F\) containing \(s_1\) and \(s_2\).
Feasibility implies \(t\notin U\). Every graph edge in
\(\delta_G(U)\) must belong to \(F\); otherwise it would join \(U\)
to another component of \(G-F\). Therefore,
\[
\cutG{U}\le c(F).
\]
Taking the minimum over feasible edge sets and feasible source sides
proves the equality.
\end{proof}

\begin{remark}[Feasibility test]
\label{rem:feasibility-test}
The instance is feasible if and only if \(s_1\) and \(s_2\) are
connected in \(G-t\). This condition can be checked before running the
approximation algorithm.
\end{remark}

\section{Preprocessing Zero-Cost Edges}
\label{sec:zero-cost}

The decomposition-tree theorem is most commonly stated for strictly
positive edge capacities. We reduce the nonnegative-cost case to that
setting without changing the asymptotic approximation guarantee.

Let
\[
E^+
=
\{e\in E:c(e)>0\},
\]
and let
\[
G^+=(V,E^+).
\]
Denote the connected component of a vertex \(v\) in \(G^+\) by
\(D(v)\).

Construct a quotient graph \(Z\) whose vertices are the connected
components of \(G^+\). Two distinct quotient vertices are adjacent
whenever a zero-cost edge of \(G\) joins their corresponding
components.

\begin{lemma}[Detection of zero-cost solutions]
\label{lem:zero-optimum}
There is a feasible RPMEC solution of cost zero if and only if
\[
D(s_1)
\quad\text{and}\quad
D(s_2)
\]
are connected in
\[
Z-D(t).
\]
Moreover, when such a solution exists, one can construct it in
polynomial time.
\end{lemma}

\begin{proof}
Suppose first that \(U\) is a feasible connected source side with
\[
\cutG{U}=0.
\]
No positive-cost edge crosses \(\delta_G(U)\). Therefore, for every
connected component \(D\) of \(G^+\), either \(D\subseteq U\) or
\(D\cap U=\varnothing\). Thus \(U\) is a union of quotient vertices.

Because \(G[U]\) is connected, the quotient vertices represented in
\(U\) induce a connected subgraph of \(Z\). This subgraph contains
\(D(s_1)\) and \(D(s_2)\) and excludes \(D(t)\). Hence \(D(s_1)\) and
\(D(s_2)\) are connected in \(Z-D(t)\).

Conversely, suppose that \(D(s_1)\) and \(D(s_2)\) are connected in
\(Z-D(t)\). Let
\[
D_0,D_1,\ldots,D_r
\]
be a quotient path between them that avoids \(D(t)\), and define
\[
U=\bigcup_{j=0}^{r}D_j.
\]
Every \(D_j\) is connected through positive-cost edges, and
consecutive components on the quotient path are joined by a zero-cost
edge. Therefore, \(G[U]\) is connected.

The set \(U\) contains \(s_1,s_2\) and excludes \(t\). Since \(U\) is
a union of complete connected components of \(G^+\), no positive-cost
edge leaves \(U\). Every edge in \(\delta_G(U)\) consequently has cost
zero, proving \(\cutG{U}=0\).
\end{proof}

Assume henceforth that no zero-cost feasible solution exists. Let
\[
\Delta
=
\min\{c(e):c(e)>0\},
\]
and let \(m=|E|\). Since every feasible solution now contains at least
one positive-cost edge,
\[
\OPT\ge\Delta.
\tag{2}\label{eq:opt-at-least-delta}
\]

Define a strictly positive perturbed cost function
\[
c_\varepsilon(e)
=
\begin{cases}
c(e), & c(e)>0,\\[1mm]
\varepsilon, & c(e)=0,
\end{cases}
\qquad
\varepsilon=\frac{\Delta}{m}.
\tag{3}\label{eq:perturbed-cost}
\]

For every edge set \(F\subseteq E\),
\[
c(F)\le c_\varepsilon(F)
\le c(F)+m\varepsilon
=c(F)+\Delta.
\tag{4}\label{eq:perturbation-bound}
\]

In particular, if \(\OPT_\varepsilon\) is the optimum under
\(c_\varepsilon\), then
\[
\OPT_\varepsilon
\le
\OPT+\Delta
\le
2\OPT.
\tag{5}\label{eq:perturbed-opt}
\]

Therefore, a \(\rho\)-approximation under \(c_\varepsilon\) has
original cost at most
\[
\rho\OPT_\varepsilon
\le
2\rho\OPT.
\tag{6}\label{eq:perturbation-factor}
\]
Thus it suffices to prove the approximation theorem for strictly
positive capacities. The factor of two is absorbed into the
\(O(\log n)\) bound.

\section{Cut-Dominating Decomposition Trees}

We now state the decomposition-tree theorem used as a black box.

A capacitated decomposition tree
\[
T=(V_T,E_T,w)
\]
is a tree whose leaves are in bijection with \(V\). Internal tree
vertices need not correspond to graph vertices.

For \(A\subseteq V\), define
\[
\lambda_T(A)
=
\min
\left\{
    \sum_{e\in F}w_e:
    F\subseteq E_T
    \text{ separates all leaves of }A
    \text{ from all leaves of }V\setminus A
\right\}.
\tag{7}\label{eq:tree-cut}
\]
Thus \(\lambda_T(A)\) is the minimum tree-cut capacity realizing the
leaf partition \(A,V\setminus A\).

\begin{definition}[Cut-dominating decomposition tree]
A decomposition tree \(T\) dominates \(G\) if
\[
\cutG{A}\le\lambda_T(A)
\qquad
\text{for every }A\subseteq V.
\tag{8}\label{eq:domination}
\]
\end{definition}

\begin{theorem}[Probabilistic decomposition-tree theorem
\cite{racke2008,andersenfeige}]
\label{thm:racke}
Let \(G=(V,E,c)\) be an undirected graph with strictly positive
capacities. In polynomial time, one can construct a distribution with
polynomial support
\[
\mathcal D
=
\{(p_i,T_i):i=1,\ldots,q\},
\]
where \(p_i\ge0\), \(\sum_i p_i=1\), and each
\[
T_i=(V_i^T,E_i^T,w_i)
\]
is a capacitated decomposition tree satisfying the following
properties.

\begin{enumerate}[label=(\roman*)]
    \item The leaves of \(T_i\) are in bijection with \(V\).

    \item Every tree edge \(e\in E_i^T\) induces a bipartition
    \[
    C_{i,e},\;V\setminus C_{i,e}
    \]
    of the graph vertices through the corresponding partition of the
    tree leaves.

    \item The capacity assigned to \(e\) is
    \[
    w_i(e)=c_G(\delta_G(C_{i,e})).
    \tag{9}\label{eq:tree-edge-induced-capacity}
    \]

    \item Every tree \(T_i\) dominates every graph cut:
    \[
    c_G(\delta_G(A))
    \le
    \lambda_{T_i}(A)
    \qquad
    \text{for every }A\subseteq V.
    \tag{10}\label{eq:tree-domination}
    \]

    \item For every fixed \(A\subseteq V\),
    \[
    \sum_{i=1}^{q}p_i\lambda_{T_i}(A)
    \le
    \rho\,c_G(\delta_G(A)),
    \tag{11}\label{eq:expected-tree-distortion}
    \]
    where \(\rho=O(\log n)\).
\end{enumerate}
\end{theorem}

\begin{remark}[Required form of the black box]
\label{rem:decomposition-structure}
The proof uses more than the existence of an arbitrary dominating
tree. It uses decomposition trees whose leaves correspond to the
vertices of \(G\), so that every rooted tree edge defines a concrete
graph vertex cluster. This cluster structure is needed to form the
induced graphs \(G[C_e]\) and their connected components.
\end{remark}

\begin{remark}[Nonnegative costs]
\label{rem:nonnegative-costs}
The strictly positive hypothesis of Theorem~\ref{thm:racke} causes no loss of
generality. By Lemma~\ref{lem:zero-optimum}, an optimum-zero instance can be
recognized and solved exactly. Otherwise, the perturbation in
\eqref{eq:perturbed-cost} converts the instance to one with strictly
positive capacities and loses at most an additional factor of two.
\end{remark}

The next elementary observation explains why decomposition trees with
edge capacities as in \eqref{eq:tree-edge-induced-capacity} dominate
graph cuts.

\begin{lemma}
\label{lem:tree-dominates}
Suppose every tree edge \(e\) has capacity
\[
w_e=\cutG{C_e},
\]
where \(C_e,V\setminus C_e\) is its induced leaf partition. Then
\[
\cutG{A}\le\lambda_T(A)
\]
for every \(A\subseteq V\).
\end{lemma}

\begin{proof}
Let \(F\subseteq E_T\) be any tree-edge set separating the leaves in
\(A\) from the leaves in \(V\setminus A\).

Consider a graph edge \(uv\in\delta_G(A)\). The leaves corresponding
to \(u\) and \(v\) lie on opposite sides of the required leaf
partition. Their unique tree path must therefore contain at least one
edge \(e\in F\). Equivalently, \(uv\in\delta_G(C_e)\) for at least one
\(e\in F\). Hence every graph edge in \(\delta_G(A)\) is counted at
least once in
\[
\sum_{e\in F}\cutG{C_e}
=
\sum_{e\in F}w_e.
\]
Therefore,
\[
\cutG{A}\le\sum_{e\in F}w_e.
\]
Minimizing over feasible tree-edge sets \(F\) proves the claim.
\end{proof}

\begin{remark}
The expected guarantee in \eqref{eq:expected-tree-distortion} is needed
only for the single fixed cut \(U^\star\) corresponding to an optimal
RPMEC solution. We do not require one tree to approximate all graph
cuts from above within \(O(\log n)\) simultaneously.
\end{remark}

\section{Rooted Clusters and Component Budgets}

Fix one decomposition tree
\[
T=(V_T,E_T,w)
\]
from the support of \(\mathcal D\). Root \(T\) at the leaf
corresponding to \(t\).

For each tree edge \(e\), let
\[
C_e\subseteq V\setminus\{t\}
\]
denote the graph vertices whose leaves lie in the component of \(T-e\)
not containing the root leaf \(t\).

Let
\[
\mathcal K_e
=
\cc(G[C_e])
\]
be the collection of vertex sets of the connected components of the
induced graph \(G[C_e]\).

For every \(K\in\mathcal K_e\), define its component cost by
\[
b(K)=\cutG{K}.
\tag{12}\label{eq:component-cost}
\]

The following identity is the principal structural fact used by the
algorithm.

\begin{lemma}[Component-budget lemma]
\label{lem:component-budget}
For every tree edge \(e\),
\[
\sum_{K\in\mathcal K_e}\cutG{K}
=
\cutG{C_e}
\le
w_e.
\tag{13}\label{eq:component-budget}
\]
\end{lemma}

\begin{proof}
Distinct sets in \(\mathcal K_e\) are connected components of the
induced graph \(G[C_e]\). Thus no graph edge has endpoints in two
different members of \(\mathcal K_e\).

It follows that every edge leaving a component
\(K\in\mathcal K_e\) also leaves \(C_e\). Conversely, every edge
leaving \(C_e\) leaves exactly one connected component of \(G[C_e]\).
Therefore,
\[
\sum_{K\in\mathcal K_e}\cutG{K}
=
\cutG{C_e}.
\]

Since cutting the single tree edge \(e\) separates the leaves of
\(C_e\) from the leaves of \(V\setminus C_e\),
\[
\lambda_T(C_e)\le w_e.
\]
Tree domination gives
\[
\cutG{C_e}\le\lambda_T(C_e).
\]
Combining the two inequalities proves \(\cutG{C_e}\le w_e\).
\end{proof}

\section{The Auxiliary Component Graph}

For the fixed rooted decomposition tree \(T\), construct a
node-weighted auxiliary graph \(H_T\).

For every tree edge \(e\in E_T\) and every component
\(K\in\mathcal K_e\), create an auxiliary node \(x_{e,K}\). Assign it
node weight
\[
b(x_{e,K})=\cutG{K}.
\tag{14}\label{eq:aux-weight}
\]

Two component nodes \(x_{e,K}\) and \(x_{f,L}\) are adjacent whenever
at least one of the following holds:
\[
K\cap L\neq\varnothing,
\tag{15}\label{eq:intersection-adjacency}
\]
or
\[
E_G(K,L)\neq\varnothing,
\tag{16}\label{eq:edge-adjacency}
\]
where \(E_G(K,L)\) denotes the set of graph edges with one endpoint in
\(K\) and the other in \(L\).

Add two zero-weight auxiliary terminals \(a\) and \(b\). Connect \(a\)
to every component node whose component contains \(s_1\), and connect
\(b\) to every component node whose component contains \(s_2\).

Let \(Q_T\) be a minimum node-weighted \(a\)-\(b\) path in \(H_T\).
Define
\[
W_T
=
\bigcup_{x_{e,K}\in V(Q_T)}K.
\tag{17}\label{eq:WT}
\]
The candidate returned for tree \(T\) is
\[
F_T=\delta_G(W_T).
\tag{18}\label{eq:FT}
\]

A node-weighted shortest path can be computed by splitting every
component node into an entrance node and an exit node connected by an
edge whose weight equals the component-node weight.

\begin{remark}[Existence of an auxiliary \(a\)-\(b\) path]
\label{rem:auxiliary-path-exists}
If \(s_1\) and \(s_2\) are connected in \(G-t\), then \(a\) and \(b\)
are connected in \(H_T\).

To see this, let
\[
P=(v_0=s_1,v_1,\ldots,v_\ell=s_2)
\]
be a path in \(G-t\). For each \(v_j\), consider the tree edge
incident to the leaf corresponding to \(v_j\). Because \(v_j\neq t\),
the rooted cluster below that edge is the singleton \(\{v_j\}\), so
\(H_T\) contains a component node representing \(\{v_j\}\).

For each \(j\), the graph edge \(v_jv_{j+1}\) makes the corresponding
component nodes adjacent in \(H_T\). The first node is adjacent to
\(a\), and the final node is adjacent to \(b\). Hence these nodes form
an \(a\)-\(b\) walk in \(H_T\), which contains a simple
\(a\)-\(b\) path.
\end{remark}

\begin{algorithm}[t]
\caption{Component expansion for one decomposition tree}
\label{alg:single-tree}
\begin{algorithmic}[1]
\Require Graph \(G=(V,E,c)\), terminals \(s_1,s_2,t\), decomposition
tree \(T\)
\Ensure A feasible RPMEC cut \(F_T\)

\State Root \(T\) at the leaf corresponding to \(t\)
\State Initialize an empty node-weighted graph \(H_T\)

\For{each tree edge \(e\in E_T\)}
    \State Let \(C_e\) be the leaf set below \(e\), away from \(t\)
    \State Compute \(\mathcal K_e=\cc(G[C_e])\)
    \For{each \(K\in\mathcal K_e\)}
        \State Add auxiliary node \(x_{e,K}\) with weight \(\cutG{K}\)
    \EndFor
\EndFor

\For{each pair \(x_{e,K},x_{f,L}\)}
    \If{\(K\cap L\neq\varnothing\) or \(E_G(K,L)\neq\varnothing\)}
        \State Add edge \(x_{e,K}x_{f,L}\) to \(H_T\)
    \EndIf
\EndFor

\State Add zero-weight auxiliary nodes \(a,b\)
\State Join \(a\) to all \(x_{e,K}\) with \(s_1\in K\)
\State Join \(b\) to all \(x_{e,K}\) with \(s_2\in K\)

\State Compute a minimum node-weighted \(a\)-\(b\) path \(Q_T\)
\State \(W_T\gets\bigcup_{x_{e,K}\in V(Q_T)}K\)
\State \Return \(F_T=\delta_G(W_T)\)
\end{algorithmic}
\end{algorithm}

\section{Feasibility and Cost of a Tree Candidate}

We first establish that the candidate generated from every tree is
feasible.

\begin{lemma}[Feasibility]
\label{lem:feasibility}
The set \(W_T\) satisfies
\[
s_1,s_2\in W_T,
\qquad
t\notin W_T,
\qquad
G[W_T]\text{ is connected}.
\]
Consequently, \(F_T=\delta_G(W_T)\) is a feasible RPMEC solution.
\end{lemma}

\begin{proof}
The first component node of \(Q_T\) after \(a\) corresponds to a
component containing \(s_1\), while the final component node before
\(b\) corresponds to a component containing \(s_2\). Therefore,
\[
s_1,s_2\in W_T.
\]

Every component used in \(H_T\) is contained in a rooted cluster
\(C_e\subseteq V\setminus\{t\}\). Hence \(t\notin W_T\).

Every component \(K\) represented by a node of \(H_T\) is connected
in \(G\). Consecutive component nodes on \(Q_T\) correspond either to
intersecting vertex sets or to vertex sets joined by a graph edge. The
union of two such connected sets is connected.

Inductively, the union of all component sets represented on \(Q_T\) is
connected. Thus \(G[W_T]\) is connected. The claim follows from
Lemma~\ref{lem:source-side}.
\end{proof}

The auxiliary path weight upper-bounds the graph boundary of the
returned union.

\begin{lemma}[Union-bound lemma]
\label{lem:union-bound}
Let
\[
Q_T=(a,x_{e_1,K_1},\ldots,x_{e_r,K_r},b).
\]
Then
\[
\cutG{W_T}
\le
\sum_{j=1}^{r}\cutG{K_j}
=
b(Q_T).
\tag{19}\label{eq:union-bound}
\]
\end{lemma}

\begin{proof}
Consider any graph edge in
\[
\delta_G\left(\bigcup_{j=1}^{r}K_j\right).
\]
One of its endpoints belongs to at least one set \(K_j\), while its
other endpoint belongs to none of the sets \(K_1,\ldots,K_r\).
Therefore, the edge belongs to \(\delta_G(K_j)\) for at least one
\(j\).

Thus every edge contributing to the left-hand side of
\eqref{eq:union-bound} is counted at least once on the right-hand side.
Since all costs are nonnegative,
\[
\cutG{W_T}
\le
\sum_{j=1}^{r}\cutG{K_j}.
\]
The final equality follows from the definition of the auxiliary node
weights.
\end{proof}

\section{Comparison with a Preserved Graph Path}

Let \(P\) be an \(s_1\)-\(s_2\) path in \(G-t\). Define the rooted
tree-separation value
\[
\phi_T(P)
=
\min
\left\{
    \sum_{e\in F}w_e:
    F\subseteq E_T,\;
    F\text{ separates the root leaf }t
    \text{ from every leaf in }V(P)
\right\}.
\tag{20}\label{eq:phi}
\]

The next lemma is the main comparison result.

\begin{lemma}[Path-comparison lemma]
\label{lem:path-comparison}
For every \(s_1\)-\(s_2\) path \(P\subseteq G-t\),
\[
\cutG{W_T}
\le
\phi_T(P).
\tag{21}\label{eq:path-comparison}
\]
\end{lemma}

\begin{proof}
Let \(F\subseteq E_T\) be an optimum edge set in \eqref{eq:phi}.
Among all optimum sets, choose \(F\) inclusion minimal. Orient all tree
edges away from the root leaf \(t\).

\paragraph{Step 1: \(F\) is an antichain.}
Suppose that \(e,f\in F\), where \(f\) is a strict descendant of
\(e\). Every root-to-leaf path containing \(f\) also contains \(e\).
Therefore, deleting \(f\) from \(F\) does not restore connectivity
between \(t\) and any vertex of \(P\). This contradicts inclusion
minimality.

Thus no edge in \(F\) is a descendant of another edge in \(F\).
Consequently, the rooted clusters \(\{C_e:e\in F\}\) are pairwise
disjoint.

\paragraph{Step 2: the rooted clusters cover the graph path.}
For every \(v\in V(P)\), the unique tree path from the root leaf \(t\)
to the leaf \(v\) contains at least one edge of \(F\). Thus \(v\)
belongs to \(C_e\) for at least one \(e\in F\). Since the clusters are
pairwise disjoint, each \(v\in V(P)\) belongs to exactly one rooted
cluster \(C_e\) with \(e\in F\).

\paragraph{Step 3: the selected components define an auxiliary walk.}
For each \(e\in F\), define
\[
\mathcal S_e
=
\{K\in\mathcal K_e:K\cap V(P)\neq\varnothing\},
\]
and let
\[
\mathcal S
=
\bigcup_{e\in F}\mathcal S_e.
\]

Write
\[
P=(v_0=s_1,v_1,\ldots,v_\ell=s_2).
\]
By the preceding step, every \(v_j\) lies in a unique rooted cluster
\(C_{e_j}\), where \(e_j\in F\). Let \(K_j\) be the unique connected
component of \(G[C_{e_j}]\) containing \(v_j\).

Consider the sequence
\[
a,\;
x_{e_0,K_0},\;
x_{e_1,K_1},\;
\ldots,\;
x_{e_\ell,K_\ell},\;b.
\tag{22}\label{eq:auxiliary-walk}
\]
The first component contains \(s_1\), so
\(a x_{e_0,K_0}\in E(H_T)\). Similarly, the final component contains
\(s_2\), so \(x_{e_\ell,K_\ell}b\in E(H_T)\).

For every \(j\in\{0,\ldots,\ell-1\}\), either
\(K_j=K_{j+1}\), or the graph edge \(v_jv_{j+1}\) has one endpoint in
\(K_j\) and the other in \(K_{j+1}\). In the latter case,
\[
x_{e_j,K_j}x_{e_{j+1},K_{j+1}}
\in E(H_T).
\]

After suppressing consecutive repetitions, the sequence in
\eqref{eq:auxiliary-walk} is therefore an \(a\)-\(b\) walk in \(H_T\).
Every walk in an undirected graph contains a simple path between its
endpoints. Let \(Q'\) be a simple \(a\)-\(b\) path obtained from this
walk by deleting closed subwalks.

\paragraph{Step 4: charge the auxiliary path to the tree cut.}
Since \(Q_T\) is a minimum node-weighted \(a\)-\(b\) path,
\[
b(Q_T)\le b(Q').
\]
Because all auxiliary node weights are nonnegative and every node of
\(Q'\) corresponds to a component in some \(\mathcal S_e\),
\[
b(Q')
\le
\sum_{e\in F}\sum_{K\in\mathcal S_e}\cutG{K}.
\]
Since \(\mathcal S_e\subseteq\mathcal K_e\),
\[
b(Q_T)
\le
\sum_{e\in F}\sum_{K\in\mathcal K_e}\cutG{K}.
\]
By Lemma~\ref{lem:component-budget},
\[
\sum_{K\in\mathcal K_e}\cutG{K}\le w_e.
\]
Consequently,
\[
b(Q_T)
\le
\sum_{e\in F}w_e
=
\phi_T(P).
\]
Finally, Lemma~\ref{lem:union-bound} gives
\[
\cutG{W_T}
\le b(Q_T)
\le\phi_T(P).
\]
\end{proof}

\section{Approximation Guarantee}

We now prove the main theorem.

\begin{theorem}
\label{thm:main}
There is a polynomial-time algorithm for weighted undirected
three-terminal RPMEC with nonnegative rational edge costs that returns
a solution of cost at most \(O(\log n)\OPT\).
\end{theorem}

\begin{proof}
We first assume that all edge costs are strictly positive. Let
\(U^\star\) be an optimum connected source side in
\eqref{eq:source-side}. Thus
\[
s_1,s_2\in U^\star,
\qquad
t\notin U^\star,
\qquad
G[U^\star]\text{ is connected},
\]
and
\[
\cutG{U^\star}=\OPT.
\]

Since \(G[U^\star]\) is connected, it contains an
\(s_1\)-\(s_2\) path
\[
P^\star\subseteq G[U^\star].
\]
Because \(t\notin U^\star\), this path lies in \(G-t\).

Fix a decomposition tree \(T_i\) in the support of \(\mathcal D\). A
tree-edge set separating all leaves of \(U^\star\) from all leaves of
\(V\setminus U^\star\) separates, in particular, the root leaf
\(t\in V\setminus U^\star\) from every leaf in
\(V(P^\star)\subseteq U^\star\). Therefore,
\[
\phi_{T_i}(P^\star)
\le
\lambda_{T_i}(U^\star).
\tag{23}\label{eq:phi-lambda}
\]

Applying Lemma~\ref{lem:path-comparison} to \(P^\star\) gives
\[
c_G(\delta_G(W_{T_i}))
\le
\phi_{T_i}(P^\star)
\le
\lambda_{T_i}(U^\star).
\tag{24}\label{eq:tree-candidate-bound}
\]

Multiply \eqref{eq:tree-candidate-bound} by \(p_i\) and sum over the
support:
\[
\sum_{i=1}^{q}p_i c_G(\delta_G(W_{T_i}))
\le
\sum_{i=1}^{q}p_i\lambda_{T_i}(U^\star).
\]
By the expected-distortion guarantee in Theorem~\ref{thm:racke},
\[
\sum_{i=1}^{q}p_i\lambda_{T_i}(U^\star)
\le
\rho\,c_G(\delta_G(U^\star))
=
\rho\OPT,
\]
where \(\rho=O(\log n)\).

Therefore, at least one support tree \(T_i\) satisfies
\[
c_G(\delta_G(W_{T_i}))
\le
\rho\OPT.
\]
The algorithm evaluates every tree in the polynomial-size support and
returns the cheapest candidate. Hence its output satisfies
\[
\ALG\le\rho\OPT=O(\log n)\OPT.
\]

For an instance with nonnegative costs,
Lemma~\ref{lem:zero-optimum} first detects and solves exactly the case
\(\OPT=0\). Otherwise, apply the strictly positive analysis to the
perturbed cost function \(c_\varepsilon\). Let \(F\) be the returned
solution. By \eqref{eq:perturbation-factor},
\[
c(F)
\le
c_\varepsilon(F)
\le
\rho\OPT_\varepsilon
\le
2\rho\OPT.
\]
Since \(\rho=O(\log n)\), the approximation ratio remains
\(O(\log n)\).
\end{proof}

The complete algorithm is summarized below.

\begin{algorithm}[t]
\caption{\(O(\log n)\)-approximation for weighted undirected RPMEC}
\label{alg:full}
\begin{algorithmic}[1]
\Require Graph \(G=(V,E,c)\), terminals \(s_1,s_2,t\)
\Ensure A feasible RPMEC cut, or a declaration of infeasibility

\If{\(s_1\) and \(s_2\) are disconnected in \(G-t\)}
    \State Report that the instance is infeasible
\EndIf

\State Construct the positive-edge component quotient \(Z\)
\If{\(D(s_1)\) and \(D(s_2)\) are connected in \(Z-D(t)\)}
    \State Construct the zero-cost source side \(U\) as in
    Lemma~\ref{lem:zero-optimum}
    \State \Return \(\delta_G(U)\)
\EndIf

\State Set
\(\Delta\gets\min\{c(e):c(e)>0\}\) and
\(\varepsilon\gets\Delta/|E|\)
\State Define \(c_\varepsilon\) as in \eqref{eq:perturbed-cost}
\State Construct the polynomial-support decomposition
\(\mathcal D=\{(p_i,T_i):i=1,\ldots,q\}\)
for \((G,c_\varepsilon)\) using Theorem~\ref{thm:racke}

\State \(F_{\mathrm{best}}\gets E\)
\State \(C_{\mathrm{best}}\gets+\infty\)

\For{\(i=1,\ldots,q\)}
    \State Run Algorithm~\ref{alg:single-tree} on \((G,c_\varepsilon,T_i)\)
    to obtain \(F_i\)
    \If{\(c(F_i)<C_{\mathrm{best}}\)}
        \State \(F_{\mathrm{best}}\gets F_i\)
        \State \(C_{\mathrm{best}}\gets c(F_i)\)
    \EndIf
\EndFor

\State \Return \(F_{\mathrm{best}}\)
\end{algorithmic}
\end{algorithm}

\section{Running Time}

Let \(N_T=|E_T|\) for one support tree \(T\).

For every tree edge \(e\), the induced graph \(G[C_e]\) has at most
\(n\) connected components. Therefore, the auxiliary graph contains
at most \(nN_T\) component nodes.

All connected components of \(G[C_e]\) can be computed in polynomial
time. Their boundary costs can be obtained by scanning graph edges. The
auxiliary adjacency relation can be constructed naively in polynomial
time by checking component intersections and graph edges.

A more efficient implementation uses vertex-to-component incidence
lists, a scan of every graph edge to generate component adjacencies,
and duplicate-edge removal.

The minimum node-weighted auxiliary path is computable in polynomial
time by node splitting followed by Dijkstra's algorithm, since all
weights are nonnegative.

The decomposition theorem supplies polynomially many trees, each of
polynomial size. Therefore, evaluating every support tree and returning
the best candidate takes polynomial time.

\section{A Randomized Sampling Variant}

When the decomposition is accessed through a sampling procedure rather
than an explicit support, the analysis gives
\[
\mathbb E_T[\cutG{W_T}]
\le
\rho\OPT.
\]
By Markov's inequality,
\[
\Pr\left[\cutG{W_T}>2\rho\OPT\right]
\le
\frac12.
\]
After sampling
\[
r
=
\left\lceil\log_2\frac{1}{\eta}\right\rceil
\]
independent trees and returning the cheapest candidate, the algorithm
returns a solution of cost at most \(2\rho\OPT\) with probability at
least \(1-\eta\).

The explicit polynomial-support version is preferable for the main
theorem because it gives a deterministic choice of the best candidate
once the decomposition distribution has been constructed.

\section{Why Component Expansion Is Necessary}

A direct tree-based approach would select a connected subtree or a
rooted tree cluster containing \(s_1\) and \(s_2\), and then interpret
its leaves as a graph vertex set. This is not sufficient: the induced
graph on those leaves may be disconnected.

The component-expansion construction resolves this difficulty without
increasing the tree charge. For every rooted cluster \(C_e\),
\[
\sum_{K\in\cc(G[C_e])}\cutG{K}
=
\cutG{C_e}
\le w_e.
\]
Thus splitting a disconnected tree cluster into its graph-connected
components preserves the total available budget.

The auxiliary graph then restores global connectivity by finding a
minimum-cost sequence of components whose union joins \(s_1\) to
\(s_2\). Importantly, this construction does not require any individual
rooted cluster to be connected in \(G\).

This also explains why the approach avoids the fractional-path dilution
that arises in polymatroid or coupled-flow relaxations. The auxiliary
solution commits to actual connected graph components and charges their
full graph boundaries. A shared bottleneck cannot be fractionally
divided among several alternative paths.

\section{Extensions and Limitations}

\paragraph{More than two protected terminals.}
Suppose a set \(R\) of protected terminals must remain connected while
being separated from \(t\). The component nodes can again be used to
construct an auxiliary node-weighted Steiner-tree instance connecting
the components containing terminals in \(R\). This suggests a
polylogarithmic approximation, although the resulting factor and
comparison proof require separate analysis.

\paragraph{Directed graphs.}
The proof relies fundamentally on undirected graph cuts, undirected
connected components, and undirected cut-dominating decomposition
trees. It does not extend directly to directed RPMEC.

\paragraph{Implementation of the tree distribution.}
The approximation theorem uses the standard polynomial-time
construction of a polynomial-support capacity-tree distribution. For
practical implementations, approximate or sampled decomposition trees
may be preferable. Their empirical quality for RPMEC remains an
interesting experimental question.

\section{Conclusion}

We presented an \(O(\log n)\)-approximation algorithm for weighted
undirected three-terminal RPMEC.

The algorithm combines probabilistic cut-dominating decomposition trees
with a connected-component expansion. For every rooted tree cluster,
its connected components in the original graph have total boundary
cost no larger than the capacity of the corresponding tree edge. A
minimum node-weighted path through these components produces a
connected feasible RPMEC source side.

The approximation proof compares this auxiliary path with a tree cut
separating an optimal preserved \(s_1\)-\(s_2\) path from \(t\). The
expected \(O(\log n)\) distortion of the decomposition-tree
distribution then yields the claimed guarantee.

The component-expansion principle may be useful more generally for cut
problems in which the retained side must satisfy an internal
connectivity constraint that is not preserved directly by a tree
embedding.

\bibliographystyle{plain}
\bibliography{references}

\end{document}